\documentclass[reprint, floatfix, superscriptaddress, amsmath, amssymb, longbibliography]{revtex4-2}

\pdfoutput=1
\usepackage[utf8]{inputenc}
\usepackage[T1]{fontenc}

\usepackage{graphicx}
\usepackage{tikz}
\usetikzlibrary{patterns,calc,arrows.meta}

\usepackage{mathtools}
\usepackage{physics}
\usepackage{braket}
\usepackage{amsthm}
\usepackage{array}
\usepackage{booktabs}
\usetikzlibrary{calc,patterns,arrows.meta,positioning,decorations.pathmorphing,decorations.markings,shapes.geometric,fit,backgrounds}

\usepackage{enumitem}

\usepackage{url}
\usepackage{hyperref}
\usepackage{cleveref}

\newtheorem{theorem}{Theorem}[section]

\newtheorem{proposition}[theorem]{Proposition}

\newcommand{\domain}{\mathcal{D}}
\newcommand{\Lz}{L_z}

\newcommand{\proj}[1]{|#1\rangle\langle#1|}
\newcommand{\Hint}{H_{\text{int}}}

\begin{document}
\raggedbottom

\title{\textbf{A Boundary Condition Perspective on Circuit QED Dispersive Readout}}

\author{Mustafa Bakr}
\email{mustafa.bakr@physics.ox.ac.uk}
\affiliation{Clarendon Laboratory, Department of Physics, University of Oxford}

\begin{abstract}
Boundary conditions in confined geometries and measurement interactions in quantum mechanics share a common structural role: both select a preferred basis by determining which states are compatible with the imposed constraint. This paper develops this perspective for circuit QED dispersive readout through a first-principles derivation starting from the circuit Lagrangian. The transmon qubit terminating a transmission line resonator provides a frequency-dependent boundary condition whose pole structure encodes the qubit's transition frequencies; different qubit states yield different resonator frequencies. Two approximations, linear response and a pole-dominated expansion valid near resonance, reduce the boundary function to a rational form in the Sturm-Liouville eigenparameter. The extended Hilbert space of the Fulton-Walter spectral theory then provides a framework for the dressed-mode eigenvalue problem conditional on the qubit state. The dispersive shift and vacuum Rabi splitting emerge from the transcendental eigenvalue equation, with the residues determined by matching to the splitting: $\delta_{ge} = 2Lg^2\omega_q^2/v^4$, where $g$ is the vacuum Rabi coupling. A level repulsion theorem guarantees that no dressed mode frequency coincides with a transmon transition. For two qubits with matched dispersive shifts, odd-parity states become frequency-degenerate; true parity-only measurement requires engineered suppression of linear dispersive terms.
\end{abstract}

\maketitle

\section{Introduction}

The quantum measurement problem admits several distinct formulations~\cite{Schlosshauer2007}. One can ask what determines which observable is being measured, why measurement outcomes correspond to eigenvalues, what accounts for the Born rule probabilities, and why observers report single definite outcomes. This paper addresses primarily the first question by developing a geometric perspective that draws on the physics of boundary conditions in confined quantum systems.

Consider a particle constrained to move on a spherical wedge with perfectly reflecting walls at $\phi = 0$ and $\phi = \Phi$~\cite{BakrAmari2025SphericalWedge}. The Dirichlet boundary conditions $\psi(0) = \psi(\Phi) = 0$ determine the allowed wavefunctions: they must be of the form $\sin(n\pi\phi/\Phi)$ for positive integers $n$. These boundary conditions also have consequences for the observable structure: the angular momentum operator $\Lz$, which is self-adjoint on the full sphere, fails to preserve the constrained domain and becomes merely symmetric on the wedge~\cite{ReedSimon1975}. We argue that the measurement Hamiltonian plays an analogous structural role in determining which observable is measured: the interaction Hamiltonian determines which observable's eigenstates are operationally stable, just as Dirichlet conditions determine which modes are compatible with the wedge geometry. We emphasize that this analogy operates at the level of functional role, not procedural identity. The wedge imposes a kinematic constraint that restricts the Hilbert space from the outset, while the measurement interaction imposes a dynamical constraint whose consequences emerge through environmental coupling. The mathematical formulations differ accordingly, i.e., domain restriction in one case, spectral extension in the other, but both encode the same structural principle: a constraint determines which states are compatible with the imposed boundary, and operators that mix compatible with incompatible states fail to be well-defined observables on the constrained system.

This paper develops three claims. The first claim is that the dispersive readout problem emerges naturally from a first-principles analysis of a transmon-terminated transmission line resonator. Starting from the circuit Lagrangian for the coupled system, we derive the wave equation in the bulk and the nonlinear boundary condition at the transmon from the variational principle. Linearization yields a frequency-dependent boundary condition; quantization of the transmon via the Kubo formula for linear response produces a state-dependent boundary function with poles at the qubit transition frequencies. In the dispersive regime where frequency shifts are small compared to the resonator frequency, this boundary function becomes approximately rational in the eigenparameter $\lambda = \omega^2/v^2$, placing the problem within the Fulton-Walter framework for Sturm-Liouville problems with eigenparameter-dependent boundary conditions~\cite{Fulton1977, Walter1973, Binding1994}. The extended Hilbert space $\mathcal{H} = L^2 \oplus \mathbb{C}^M$ of the spectral theory provides a mathematical structure for the resonator-qubit decomposition. The dispersive shift $\chi = g^2\alpha/[\Delta(\Delta+\alpha)]$ emerges from this framework, with the residue formula $\delta = 2\omega_r|\omega_q||Q_{ge}|^2/(\hbar v Z_0)$ derived explicitly in terms of circuit parameters. The interaction Hamiltonian $\Hint = \hbar\chi\sigma_z\hat{n}$ satisfies $[\Hint, \sigma_z] = 0$, which is why dispersive readout measures $\sigma_z$ and not $\sigma_x$~\cite{Blais2004, Wallraff2004}. The $\sigma_z$ eigenstates are the operationally stable states under dispersive measurement, just as Dirichlet boundary conditions exclude cosine modes from the physical Hilbert space of a confined geometry. This structural principle, that observables commuting with a dominant Hamiltonian become operationally stable, connects to foundational work on protective measurement~\cite{AharonovAnandanVaidman1993} and recent developments in Hamiltonian stabilization for quantum error correction~\cite{Cortinas2025}.

The second claim concerns the spectral structure and multimode physics. An interlacing theorem, proved using the monotonicity properties of the resonator function $G(\lambda) = \sqrt{\lambda}\cot(\sqrt{\lambda}L)$ and the boundary function $F(\lambda)$, guarantees exactly one dressed eigenvalue between each consecutive pair of poles. This provides a foundation for level repulsion and yields the vacuum Rabi splitting $2g$ at resonance through explicit calculation of $G'(\lambda)$ at the crossing point. The multimode extension reveals that the mode-dependent coupling scales as $g_n \propto \sqrt{\omega_n}$, leading to an ultraviolet divergence in the Lamb shift that requires renormalization: the ``bare'' qubit frequency in the Hamiltonian differs from the physical (measured) frequency by a formally infinite correction. The dispersive shift, by contrast, remains finite because its sum over modes converges.

The third claim extends this perspective to multiple qubits and error correction. For two qubits coupled to a shared resonator, the dispersive Hamiltonian commutes with parity $P = \sigma_1^z\sigma_2^z$, making parity a quantum non-demolition (QND) observable within the dispersive approximation. When dispersive shifts are matched ($\chi_1 = \chi_2$), the resonator becomes insensitive to which odd-parity state is occupied ($|ge\rangle$ versus $|eg\rangle$), but $|gg\rangle$ and $|ee\rangle$ remain distinguishable. Achieving a true parity-only response, where all states within each parity sector produce identical resonator frequencies, requires engineered cancellation of the linear dispersive terms~\cite{Royer2018, DiVincenzo2013}. The boundary condition framework provides conceptual insight into these schemes: the apparatus is designed so that parity eigenstates become the operationally stable configurations. The stabilizer conditions $S_i\ket{\psi} = +\ket{\psi}$ defining a quantum error-correcting code space~\cite{Gottesman1997} have the same mathematical structure as boundary conditions defining allowed modes: both are constraints that select a subspace of compatible states. We demonstrate this correspondence explicitly for parity stabilizers, where the physical implementation via dispersive coupling provides a concrete realization. However, $X$-type stabilizers require different physical implementations, typically ancilla qubits and sequences of two-qubit gates, and the correspondence for general stabilizer codes remains structural rather than derived from circuit physics.

This framework addresses how the measurement Hamiltonian selects the eigenbasis, provides explicit first-principles analysis for single-qubit dispersive readout and two-qubit parity measurement, and connects these to the mathematical structure of stabilizer codes. The framework does not address Bell nonlocality and EPR correlations~\cite{Bell1964}, why observers experience definite outcomes~\cite{Wallace2012}, the ontological status of branches in the post-measurement state~\cite{Everett1957}, or the origin of Born rule probabilities~\cite{Gleason1957}. These questions, while important, lie outside the scope of the present analysis. We emphasize from the outset that ``basis selection'' and ``outcome selection'' are distinct problems. Basis selection asks: given a measurement apparatus, which observable does it measure? Outcome selection asks: given that observable, why does a particular eigenvalue occur? This paper addresses the first question through physical analysis derived from circuit physics. The second question remains contested and lies outside our scope. The framework takes the classicality of the measurement apparatus as a prerequisite---established through environmental decoherence, as analyzed in Section~\ref{sec:classical_bc}---and investigates the structural consequences for basis selection. Paz and Zurek have shown that in the weak-dissipation regime, the system's self-Hamiltonian determines the einselected basis independently of environmental details~\cite{PazZurek1999}. The boundary condition perspective makes this concrete: the structure of the interaction Hamiltonian, not the specifics of decoherence, determines which observable is measured.

Throughout this paper, we use the term ``branch'' in a precise mathematical sense: a branch is a term in the decoherence-diagonal density matrix~\cite{Zurek2003, Joos1985}. If the total state after decoherence takes the form $\rho \approx \sum_i p_i \rho_i^{(S)} \otimes \proj{\phi_i^{(E)}}$, where the environment states $\ket{\phi_i^{(E)}}$ are approximately orthogonal, then each term constitutes a branch. This definition is purely mathematical and carries no ontological commitment. It is compatible with the Many-Worlds interpretation~\cite{Everett1957}, Copenhagen-like interpretations, relational quantum mechanics~\cite{Rovelli1996}, and QBism~\cite{Fuchs2014}. Figure~\ref{fig:mode_selection} illustrates the structural parallel across these three settings: Dirichlet boundary conditions selecting angular modes in a wedge, dispersive coupling selecting $\sigma_z$ eigenstates in single-qubit readout, and matched dispersive shifts selecting parity eigenstates in two-qubit measurement.

\begin{figure*}[t]
\centering
\begin{tikzpicture}[scale=0.9]

\begin{scope}[shift={(0,0)}]
    \node[font=\bfseries] at (1.5,3.5) {(A)};
    
    \fill[blue!10] (0,0) -- (3.2,0) arc (0:55:3.2) -- cycle;
    
    \draw[very thick] (0,0) -- (3.4,0);
    \draw[very thick] (0,0) -- (1.95,2.8);
    
    \node[below, font=\small] at (1.7,-0.2) {$\phi=0$};
    \node[above left, font=\small] at (1.7,2.7) {$\phi=\Phi$};
    
    \draw[->, thick] (1.0,0) arc (0:55:1.0);
    \node[font=\small] at (1.25,0.65) {$\Phi$};
    
    \fill (0,0) circle (1.5pt);
\end{scope}

\begin{scope}[shift={(6,0)}]
    \node[font=\bfseries] at (1.8,3.5) {(B)};
    
    \node[font=\small] at (0,3.0) {R};
    \draw[thick, gray] (-0.5,0.6) -- (0.5,0.6);
    \draw[thick, gray] (-0.5,1.2) -- (0.5,1.2);
    \draw[thick, gray] (-0.5,1.8) -- (0.5,1.8);
    \draw[thick, gray] (-0.5,2.4) -- (0.5,2.4);
    \node[left, font=\tiny, gray] at (-0.55,0.6) {$|0\rangle$};
    \node[left, font=\tiny, gray] at (-0.55,1.2) {$|1\rangle$};
    \node[left, font=\tiny, gray] at (-0.55,1.8) {$|2\rangle$};
    
    \node[font=\small] at (3.5,3.0) {Qubit};
    
    \draw[very thick, blue] (2.8,0.8) -- (4.2,0.8);
    \node[right, font=\small, blue] at (4.25,0.8) {$|g\rangle$};
    
    \draw[very thick, orange] (2.8,2.2) -- (4.2,2.2);
    \node[right, font=\small, orange] at (4.25,2.2) {$|e\rangle$};
    
    \draw[thick, <->, decorate, decoration={snake, amplitude=1.5pt, segment length=6pt}] 
        (0.6,1.5) -- (2.6,1.5);
    
    \node[font=\small] at (1.6,1.0) {$\hbar\chi\sigma_z\hat{n}$};
\end{scope}

\begin{scope}[shift={(12,0)}]
    \node[font=\bfseries] at (1.3,3.5) {(C)};
    
    \node[font=\small] at (1.3,3.0) {R};
    \draw[thick, gray] (0.8,2.2) -- (1.8,2.2);
    \draw[thick, gray] (0.8,2.6) -- (1.8,2.6);
    
    \node[font=\small] at (-0.2,1.5) {$Q_1$};
    \draw[thick, blue] (-0.7,0.7) -- (0.3,0.7);
    \draw[thick, orange] (-0.7,1.2) -- (0.3,1.2);
    
    \node[font=\small] at (2.8,1.5) {$Q_2$};
    \draw[thick, blue] (2.3,0.7) -- (3.3,0.7);
    \draw[thick, orange] (2.3,1.2) -- (3.3,1.2);
    
    \draw[thick, <->, decorate, decoration={snake, amplitude=1.2pt, segment length=5pt}] 
        (0.35,1.15) -- (0.85,2.1);
    \draw[thick, <->, decorate, decoration={snake, amplitude=1.2pt, segment length=5pt}] 
        (1.75,2.1) -- (2.25,1.15);
    
    \node[font=\tiny] at (0.35,1.75) {$\chi_1$};
    \node[font=\tiny] at (2.25,1.75) {$\chi_2$};
    
    \node[font=\small] at (1.3,0.3) {$\chi_1 = \chi_2$};
\end{scope}

\end{tikzpicture}

\caption{Mode selection by constraints. (A)~Spherical wedge with Dirichlet boundary conditions $\psi(0)=\psi(\Phi)=0$. 
Allowed modes are $\sin(n\pi\phi/\Phi)$, which vanish at both walls; 
cosine modes violate the boundary at $\phi=0$ and are excluded. (B)~Single-qubit dispersive readout. A resonator R couples to a qubit 
via $\Hint = \hbar\chi\sigma_z\hat{n}$; the qubit state shifts the 
resonator frequency, making $|g\rangle$ and $|e\rangle$ distinguishable. (C)~Two-qubit parity measurement. When $\chi_1 = \chi_2$, the odd-parity states $\lvert ge \rangle$ and $\lvert eg \rangle$ produce identical resonator frequencies, while the even-parity states $\lvert gg \rangle$ and $\lvert ee \rangle$ remain distinguishable. This degeneracy alone does not achieve true parity-only measurement, which requires additional engineering to cancel the linear dispersive terms (see Section~\ref{sec:parity}).}
\label{fig:mode_selection}
\end{figure*}
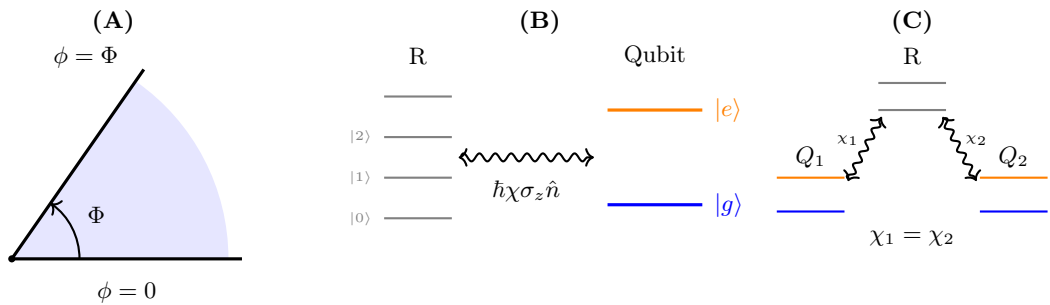

The paper proceeds as follows. Section~\ref{sec:wedge} develops the physics of the spherical wedge and shows how boundary conditions constrain admissible observables. Section~\ref{sec:classical_bc} examines how decoherence establishes the classicality of boundaries. Section~\ref{sec:measurement} applies this perspective to quantum measurement. Section~\ref{sec:dynamic_bc} develops the dispersive readout problem from first principles, starting from the circuit Lagrangian, deriving the boundary condition at the transmon, and obtaining the state-dependent boundary function via quantum linear response. Section~\ref{sec:spectral} analyzes the spectral structure through Sturm-Liouville theory, proves the interlacing theorem, derives the vacuum Rabi splitting, and addresses multimode physics including ultraviolet divergences and renormalization. Section~\ref{sec:parity} extends the framework to multiple qubits and analyzes parity measurement. Section~\ref{sec:qec} discusses the structural correspondence to quantum error correction. Section~\ref{sec:born} addresses the Born rule within standard assumptions. Section~\ref{sec:open} is explicit about what remains unexplained.

\section{The Spherical Wedge and Observable Structure}
\label{sec:wedge}

Consider a particle constrained to move on a portion of a sphere, specifically a wedge defined by $\phi \in [0, \Phi]$ with the polar angle $\theta$ unrestricted. We impose Dirichlet boundary conditions: the wavefunction must vanish at both walls,
\begin{equation}
\psi(\theta, 0) = \psi(\theta, \Phi) = 0 \quad \text{for all } \theta.
\label{eq:dirichlet}
\end{equation}
The azimuthal part of the wavefunction must satisfy these conditions. The solutions are $\sin(\mu\phi)$ where $\mu\Phi = n\pi$ for positive integers $n$, giving
\begin{equation}
\mu_n = \frac{n\pi}{\Phi}, \quad n = 1, 2, 3, \ldots
\label{eq:mu_quantization}
\end{equation}
On the full sphere with periodic boundary conditions, the angular momentum quantum numbers are integers $m \in \mathbb{Z}$. On the wedge, the effective quantum numbers $\mu_n$ are generically non-integer.

A general state in the wedge domain is a superposition
\begin{equation}
\psi(\phi) = \sum_{n=1}^{\infty} c_n \sin\left(\frac{n\pi\phi}{\Phi}\right),
\label{eq:general_wedge_state}
\end{equation}
with arbitrary coefficients $c_n$ subject to normalization. The boundary conditions do not select a single mode; they select the \textit{basis} of allowed modes.

\subsection{The Fate of $\Lz$}

On the full sphere, the angular momentum operator $\Lz = -i\hbar\, \partial/\partial\phi$ is self-adjoint with eigenfunctions $e^{im\phi}$ and integer eigenvalues $m\hbar$. On the wedge with Dirichlet conditions, the situation changes. Consider the action of $\Lz$ on a basis function:
\begin{equation}
\Lz \sin\left(\frac{n\pi\phi}{\Phi}\right) = -i\hbar \frac{n\pi}{\Phi} \cos\left(\frac{n\pi\phi}{\Phi}\right).
\label{eq:Lz_on_sin}
\end{equation}
The cosine function does not vanish at $\phi = 0$: $\cos(0) = 1 \neq 0$. Therefore $\Lz$ maps functions satisfying the Dirichlet condition to functions that violate it.

\begin{proposition}
\label{prop:Lz_symmetric}
The operator $\Lz$ is symmetric but not self-adjoint on the Dirichlet domain $\domain = \{\psi \in H^1([0,\Phi]) : \psi(0) = \psi(\Phi) = 0\}$.
\end{proposition}

\begin{proof}
The full inner product on the spherical wedge is 
$\langle \psi | \chi \rangle = \int_0^\pi \sin\theta\, d\theta \int_0^\Phi d\phi\, 
\psi^*(\theta,\phi)\chi(\theta,\phi)$. For states of the form 
$\psi(\theta,\phi) = \Theta(\theta)\varphi(\phi)$, the $\theta$-integration 
factors out, and the relevant inner product for the azimuthal part is 
$\langle \varphi_1 | \varphi_2 \rangle_\phi = \int_0^\Phi \varphi_1^*(\phi) 
\varphi_2(\phi)\, d\phi$. We work with this reduced inner product in what follows.

For $\psi, \chi \in \domain$:
\begin{equation}
\langle \psi | \Lz \chi \rangle_\phi = -i\hbar \int_0^\Phi \psi^*(\phi) 
\frac{d\chi}{d\phi}\, d\phi = \langle \Lz \psi | \chi \rangle_\phi + 
i\hbar[\psi^* \chi]_0^\Phi.
\end{equation}
The boundary term vanishes since both $\psi$ and $\chi$ satisfy Dirichlet 
conditions, so $\Lz$ is symmetric. However, $\Lz$ is not self-adjoint because $\domain(\Lz) \subsetneq \domain(\Lz^*)$. 
The adjoint is defined on a larger domain that includes functions not satisfying 
the boundary conditions.
\end{proof}
Equivalently, the Dirichlet domain is not invariant under $\Lz$: the operator maps functions satisfying the boundary conditions to functions that violate them. On this domain, $\Lz$ is not a self-adjoint observable; making $\Lz$ self-adjoint requires changing the domain, which corresponds physically to imposing different boundary conditions.

\subsection{Operators That Preserve the Domain}
In contrast to $\Lz$, the second-order operator $\partial^2/\partial\phi^2$ preserves the Dirichlet domain. Being second-order in $\phi$, it maps sine functions to sine functions:
\begin{equation}
\frac{\partial^2}{\partial\phi^2} \sin\left(\frac{n\pi\phi}{\Phi}\right) = -\left(\frac{n\pi}{\Phi}\right)^2 \sin\left(\frac{n\pi\phi}{\Phi}\right).
\label{eq:L2_on_sin}
\end{equation}
This operator preserves the boundary conditions and hence preserves the domain, allowing us to define a self-adjoint angular operator on the wedge. The eigenvalues are $-\mu_n^2 = -(n\pi/\Phi)^2$, as shown in Eq.~\eqref{eq:L2_on_sin}. On the full sphere with periodic boundary conditions, the analogous operator $L_z^2$ has eigenvalues $m^2\hbar^2$ for integer $m$; on the wedge, the effective quantum numbers $\mu_n = n\pi/\Phi$ are generically non-integer. From a symmetry perspective, the wedge walls break the SO(2) rotational symmetry. When this symmetry is broken, the generator $\Lz$ no longer preserves the physical domain, while higher-order operators that do not generate the broken symmetry can still be defined.

This distinction between differential operators and physically admissible observables is not specific to quantum measurement. An explicit and fully classical realization occurs in wave systems with restricted domains, where symmetry generators cease to be self-adjoint while higher-order operators remain well-defined. Observable selection is enforced by self-adjointness in the restricted domain rather than by algebraic considerations alone~\cite{BakrAmari2025Trichotomy, BakrZhangAmari2025SphericalCavities}.

\section{Classical Boundaries and Decoherence}
\label{sec:classical_bc}
The wedge boundary conditions $\psi(0) = \psi(\Phi) = 0$ impose definite constraints on the allowed modes. But what makes these constraints effective? The answer lies in the classical nature of the boundary itself, which emerges from decoherence.

\subsection{Why Boundaries Must Be Classical}
The walls of a physical wedge consist of a large ensemble of atoms, 
approximately $10^{23}$, each coupled to electromagnetic and phononic 
environments. Any quantum superposition of wall configurations would 
decohere on extremely short timescales, far shorter than any relevant 
particle dynamics~\cite{Joos1985}. This rapid decoherence is what makes 
the boundary ``classical'' and therefore capable of imposing a definite 
constraint.

To see why classicality is essential, consider what would happen if we attempted to place the wedge angle in superposition. Mathematically, different wedge angles $\Phi_1$ and $\Phi_2$ define different Hilbert spaces: $\mathcal{H}_{\Phi_1} = L^2([0,\Phi_1])$ and $\mathcal{H}_{\Phi_2} = L^2([0,\Phi_2])$ with their respective Dirichlet conditions. These are distinct spaces with different dimensions and different basis functions. One cannot canonically form superpositions of vectors belonging to different Hilbert spaces without specifying an embedding or identification map between them.

To describe a system where the geometry itself is quantum mechanical, one must work in a larger Hilbert space $\mathcal{H}_{\text{total}} = \mathcal{H}_{\text{particle}} \otimes \mathcal{H}_{\text{geometry}}$ that includes the geometric configuration as a quantum degree of freedom. In this enlarged space, a state of the form
\begin{equation}
\ket{\Psi} = \alpha \ket{\psi^{(1)}}\ket{\Phi_1} + \beta \ket{\psi^{(2)}}\ket{\Phi_2}
\label{eq:enlarged_space}
\end{equation}
is well-defined, where $\ket{\psi^{(i)}}$ represents the particle state appropriate to geometry $\Phi_i$. But this larger space has its own fixed structure, we have embedded the problem in a larger arena rather than making the boundary conditions themselves quantum mechanical. The boundary conditions of the total system remain classical.

\subsection{Decoherence and Effective Classicality}

The physical resolution comes from decoherence. When the wall configuration becomes entangled with environmental degrees of freedom, the coherence between different geometric configurations is rapidly suppressed. If the wall starts in a superposition and interacts with the environment, the evolution takes the form
\begin{equation}
\bigl(\alpha\ket{\Phi_1} + \beta\ket{\Phi_2}\bigr)\ket{E_0} \longrightarrow \alpha\ket{\Phi_1}\ket{E_1} + \beta\ket{\Phi_2}\ket{E_2},
\label{eq:wall_decoherence}
\end{equation}
where the environment states satisfy $\braket{E_1}{E_2} \approx 0$ after decoherence. Tracing over the environment yields a reduced density matrix for the wall that is approximately diagonal in the configuration basis:
\begin{equation}
\rho_{\text{wall}} \approx |\alpha|^2 \proj{\Phi_1} + |\beta|^2 \proj{\Phi_2}.
\label{eq:wall_reduced}
\end{equation}
A particle interacting with such a decohered wall experiences an effectively definite geometry. The interference terms that would allow the particle to ``see'' a superposition of boundary conditions are suppressed by the factor $\braket{E_1}{E_2}$, which for macroscopic walls is exponentially small in the number of environmental degrees of freedom.

\subsection{The Parallel to Measurement}

This structure parallels the quantum measurement process, but with an important distinction. The measurement apparatus, like the wedge wall, becomes effectively classical through decoherence. However, the interaction Hamiltonian alone does not cause decoherence, an additional ingredient is required.

Consider an electron spin in a magnetic field along $z$. The interaction Hamiltonian $H = \gamma S_z B$ commutes with $S_z$, so in some sense it ``selects'' the $S_z$ basis. But a superposition $\alpha|\uparrow\rangle + \beta|\downarrow\rangle$ simply precesses; no decoherence occurs, no measurement happens. The magnetic field, despite providing an interaction Hamiltonian, is not an environment, it has no internal degrees of freedom to become entangled with, and the dynamics is fully reversible.

Decoherence requires coupling to a system with many degrees of freedom whose detailed state becomes practically inaccessible. For the wedge wall, this is the electromagnetic and phononic environment of the wall atoms. For a measurement apparatus, this is whatever carries information away irreversibly---in circuit QED, it is the transmission line into which resonator photons leak.

The parallel can now be stated precisely. For the wedge boundary, the wall atoms decohere through coupling to their environment, leading to a definite geometric configuration, which in turn produces a definite mode structure. For the measurement apparatus, the pointer states decohere through coupling to an external environment, leading to definite pointer records, which in turn produces a definite branch structure. The interaction Hamiltonian determines \emph{which} basis is selected; the environmental coupling determines \emph{that} decoherence occurs. Both ingredients are necessary. In both cases, classicality, achieved through environmental decoherence, is what enables the imposition of definite constraints.

\subsection{What Decoherence Does and Does Not Explain}

This parallel illuminates basis selection but not outcome selection. Decoherence explains why the measurement has a definite basis (the pointer states that survive decoherence) just as the classical nature of the wall explains why the wedge has definite modes (the Dirichlet eigenfunctions compatible with the wall positions). However, neither decoherence nor the geometric analogy explains why a particular outcome occurs. After decoherence, the global state of system plus apparatus plus environment takes the form
\begin{equation}
\ket{\Psi} = \sum_n c_n \ket{n}\ket{A_n}\ket{E_n},
\label{eq:global_decohered}
\end{equation}
where $\braket{E_m}{E_n} \approx \delta_{mn}$. All branches exist in this state; decoherence has not eliminated any of them. The reduced density matrix of the system appears mixed, but the global state remains a pure superposition. The question of why an observer finds themselves in one branch rather than another is not addressed by decoherence and lies outside the scope of the geometric analogy developed here. The wedge case is importantly different: the cosine modes are not ``somewhere else'' in a larger state, they simply do not exist as elements of the Hilbert space $L^2([0,\Phi])$ with Dirichlet conditions. In the wedge, incompatible modes are absent from the chosen operator domain for that boundary value problem. In measurement, incompatible coherences remain in the global state but become operationally inaccessible under realistic environment decoherence. This distinction is precisely why the geometric analogy illuminates the structure of measurement (basis selection) without resolving its deepest puzzle (outcome selection).

\section{Measurement and Observable Selection}
\label{sec:measurement}

We now apply the boundary condition perspective to quantum measurement. A complete measurement process requires two distinct physical ingredients: an interaction Hamiltonian that determines which observable is measured, and coupling to an environment that causes decoherence. The interaction Hamiltonian alone does not constitute measurement, it merely creates entanglement. We examine both ingredients in the context of circuit QED dispersive readout.

\subsection{The Interaction Hamiltonian: Basis Selection}

Consider a qubit coupled to a resonator via dispersive readout~\cite{Blais2004, Blais2021}. The interaction Hamiltonian is
\begin{equation}
\Hint = \hbar\chi\, \sigma_z \otimes \hat{n},
\label{eq:H_int}
\end{equation}
where $\chi$ is the dispersive shift and $\hat{n} = a^\dagger a$ is the resonator photon number. This Hamiltonian commutes with $\sigma_z$:
\begin{equation}
[\Hint, \sigma_z] = 0,
\label{eq:commutator_Hz}
\end{equation}
but does not commute with $\sigma_x$ or $\sigma_y$:
\begin{equation}
[\Hint, \sigma_x] = 2i\hbar\chi\, \hat{n}\, \sigma_y \neq 0.
\label{eq:commutator_Hx}
\end{equation}

The commutation relations determine which basis the measurement will select: observables that commute with $\Hint$ are compatible with the interaction, while those that do not commute will become entangled with the apparatus. For dispersive readout, $\sigma_z$ is the compatible observable.

However, the interaction Hamiltonian alone does not cause decoherence. If the qubit-resonator system were perfectly isolated, a qubit initially in superposition $\alpha\ket{g} + \beta\ket{e}$ would evolve to an entangled state
\begin{equation}
\ket{\Psi(t)} = \alpha\ket{g}\ket{\psi_g(t)} + \beta\ket{e}\ket{\psi_e(t)},
\label{eq:entangled_state}
\end{equation}
where $\ket{\psi_g(t)}$ and $\ket{\psi_e(t)}$ are resonator states that acquire different phases. This entanglement is reversible: with sufficient control, one could disentangle the qubit and resonator, restoring the original superposition. No measurement has occurred, only reversible unitary evolution.

This situation is analogous to the electron spin in a magnetic field discussed in Section~\ref{sec:classical_bc}. The Hamiltonian $H = \gamma S_z B$ commutes with $S_z$, apparently ``selecting'' the $S_z$ basis, but a spin superposition simply precesses without decoherence. The magnetic field has no internal degrees of freedom to carry away information irreversibly.

\subsection{Environmental Coupling: Decoherence}

Decoherence requires coupling to a large environment whose detailed quantum state becomes practically inaccessible. In circuit QED, this environment is the transmission line to which the resonator is coupled. The transmission line supports a continuum of electromagnetic modes; when photons leak from the resonator into this continuum, they carry phase information that distinguishes $\ket{g}$ from $\ket{e}$.

The crucial physical process is photon emission into the line. Each photon that escapes carries information about the qubit-dependent resonator frequency: photons emitted when the qubit is in $\ket{g}$ have a slightly different frequency than those emitted when the qubit is in $\ket{e}$. As photons propagate down the transmission line, they become entangled with the qubit state:
\begin{equation}
\ket{\Psi} = \alpha\ket{g}\ket{E_g} + \beta\ket{e}\ket{E_e},
\label{eq:branched_state}
\end{equation}
where $\ket{E_g}$ and $\ket{E_e}$ represent the state of the electromagnetic field in the transmission line.

The key property is that $\braket{E_g}{E_e} \to 0$ as more photons leak out. This orthogonality arises because the field states differ in their frequency content, and the distinguishability accumulates with each emitted photon. Tracing over the transmission line modes yields a reduced density matrix for the qubit that is approximately diagonal in the $\sigma_z$ basis:
\begin{equation}
\rho_S = |\alpha|^2 \proj{g} + |\beta|^2 \proj{e}.
\label{eq:decohered_qubit}
\end{equation}
This represents two effectively autonomous branches, each correlated with a distinct record in the environment. The off-diagonal coherences $\alpha^*\beta\ket{e}\bra{g}$ are suppressed by the factor $\braket{E_g}{E_e}$, which decreases exponentially with the number of distinguishing photons.

\subsection{Conditions for Coherence Preservation}

Decoherence is not inevitable; it requires sufficient environmental coupling that produces distinguishable environmental states. Several physical scenarios illustrate conditions under which off-diagonal coherences persist.

If the resonator has no coupling to a transmission line (or equivalently, infinite quality factor), no photons leak out and the qubit-resonator system undergoes reversible unitary dynamics indefinitely. Similarly, if the qubit states produce identical resonator responses, then $\ket{E_g} = \ket{E_e}$ and $\braket{E_g}{E_e} = 1$; photons leaking out carry no which-path information, and coherence is preserved. This latter scenario is realized for odd-parity states under matched dispersive shifts (Section~\ref{sec:parity}): $\ket{ge}$ and $\ket{eg}$ produce the same resonator frequency, so the coherence $\braket{ge|\rho|eg}$ remains intact. Additionally, processes occurring faster than the characteristic dephasing rate $\Gamma_\phi$ can complete before significant decoherence accumulates, and techniques such as spin echo and dynamical decoupling can reverse environmental entanglement before it becomes irreversible.

These considerations confirm that the interaction Hamiltonian and environmental coupling are distinct ingredients. The Hamiltonian determines which basis would be selected if decoherence occurs; the environmental coupling determines whether and how rapidly decoherence proceeds.

\subsection{The Structural Parallel to Boundary Conditions}

With both ingredients identified, the parallel between measurement and boundary conditions can be stated precisely. In the wedge, Dirichlet conditions determine which modes are allowed: functions must vanish at $\phi = 0$ and $\phi = \Phi$. The operator $\Lz$ maps Dirichlet-satisfying functions outside the domain; the operator $\partial^2/\partial\phi^2$ preserves the domain and remains well-defined. In measurement, the interaction Hamiltonian $\Hint$ determines which eigenstates are compatible with the measurement process: eigenstates of observables commuting with $\Hint$ remain unentangled, while superpositions become entangled with the apparatus. Operators that do not commute with $\Hint$ mix these eigenstates and become operationally inaccessible once decoherence occurs.

The wedge imposes a kinematic constraint: incompatible modes do not exist in the Hilbert space. Measurement involves dynamic constraints: the interaction Hamiltonian selects the basis, and environmental coupling causes the decoherence that makes this selection irreversible. The structural parallel lies in the role of the constraint—determining which states are operationally stable—while the mechanism differs fundamentally. In the wedge, boundary conditions eliminate cosine modes from the physical Hilbert space. In measurement, all branches continue to exist in the total quantum state; they become operationally inaccessible through environmental entanglement, not ontologically absent. The analogy thus concerns the structure of basis selection rather than the mechanism of reduction to definite outcomes.

This analogy concerns local mode selection. The Dirichlet condition acts at the boundary; $\Hint$ acts at the system-apparatus interface; decoherence occurs through local environmental coupling. The framework does not address Bell nonlocality~\cite{Bell1964}, EPR correlations, or whether measurement involves nonlocal effects.

\section{Dispersive Readout}
\label{sec:dynamic_bc}

The preceding sections developed a geometric perspective on measurement through the analogy with boundary conditions. We now show that this perspective emerges naturally from the circuit physics of a transmon-terminated transmission line resonator. Starting from the classical Lagrangian, we derive the equations of motion, obtain the boundary condition at the qubit, and show how quantization of the transmon leads to a state-dependent boundary function whose pole structure encodes the qubit's transition frequencies.

\subsection{The Circuit Lagrangian}

Consider a transmission line resonator of length $L$ with distributed inductance $\ell$ per unit length and distributed capacitance $c$ per unit length. The phase velocity is $v = 1/\sqrt{\ell c}$ and the characteristic impedance is $Z_0 = \sqrt{\ell/c}$. These parameters satisfy $\ell = Z_0/v$ and $c = 1/(Z_0 v)$. The transmission line is grounded at $x = 0$ and terminated by a transmon qubit at $x = L$.

The dynamical variable is the flux field $\Phi(x,t)$, which relates to the voltage and current through $V(x,t) = \partial\Phi/\partial t$ and $I(x,t) = -(\partial\Phi/\partial x)/\ell$. These relations follow from Faraday's law and the constitutive relation for distributed inductance.

The transmission line stores energy in both its distributed capacitance, which contributes kinetic (electric) energy $T_{\text{TL}} = \frac{1}{2}\int_0^L c\,\dot{\Phi}^2\,dx$, and its distributed inductance, which contributes potential (magnetic) energy $U_{\text{TL}} = \frac{1}{2}\int_0^L (\partial_x\Phi)^2/\ell\,dx$. The transmon at $x = L$ consists of a Josephson junction with Josephson energy $E_J$ shunted by a capacitance $C_J$. Its flux variable $\Phi_J$ is constrained to equal $\Phi(L,t)$ by the circuit topology. The transmon contributes kinetic energy $T_J = \frac{1}{2}C_J\dot{\Phi}_J^2$ from its capacitance and potential energy $U_J = -E_J\cos(2\pi\Phi_J/\Phi_0)$ from the Josephson junction, where $\Phi_0 = h/(2e)$ is the flux quantum.

The total Lagrangian of the coupled system is therefore
\begin{multline}
\mathcal{L} = \int_0^L \left[\frac{c}{2}\dot{\Phi}^2 - \frac{1}{2\ell}\left(\frac{\partial\Phi}{\partial x}\right)^2\right]dx \\
+ \frac{C_J}{2}\dot{\Phi}(L,t)^2 + E_J\cos\left(\frac{2\pi\Phi(L,t)}{\Phi_0}\right).
\label{eq:total_lagrangian}
\end{multline}
This Lagrangian describes a distributed system (the transmission line) coupled to a lumped nonlinear element (the transmon) at its boundary. The coupling arises because the transmon flux $\Phi_J$ is identified with the transmission line flux at $x = L$.

\subsection{Equations of Motion and the Boundary Condition}

The equations of motion follow from the variational principle $\delta S = 0$ where $S = \int \mathcal{L}\,dt$. Varying with respect to $\Phi(x,t)$ in the bulk ($0 < x < L$) yields the wave equation
\begin{equation}
\frac{\partial^2\Phi}{\partial t^2} = v^2\frac{\partial^2\Phi}{\partial x^2},
\label{eq:wave_eqn}
\end{equation}
which governs the propagation of electromagnetic waves along the transmission line.

At $x = 0$, we impose a short circuit to ground, giving the Dirichlet boundary condition $\Phi(0,t) = 0$. At $x = L$, varying the action with respect to $\Phi(L,t)$ and setting the total variation to zero gives
\begin{multline}
\frac{1}{\ell}\frac{\partial\Phi}{\partial x}\bigg|_{x=L} + C_J\frac{\partial^2\Phi(L,t)}{\partial t^2} \\
- \frac{2\pi E_J}{\Phi_0}\sin\left(\frac{2\pi\Phi(L,t)}{\Phi_0}\right) = 0.
\label{eq:bc_nonlinear}
\end{multline}
This equation expresses current conservation at the boundary: the current delivered by the transmission line equals the current absorbed by the transmon capacitor plus the Josephson supercurrent.

\subsection{Linearization and the Harmonic Approximation}

For small flux amplitudes satisfying $|2\pi\Phi_J/\Phi_0| \ll 1$, we linearize to obtain
\begin{equation}
\frac{1}{\ell}\frac{\partial\Phi}{\partial x}\bigg|_{x=L} + C_J\ddot{\Phi}(L,t) + \frac{\Phi(L,t)}{L_J} = 0,
\label{eq:bc_linear}
\end{equation}
where $L_J = \Phi_0^2/(4\pi^2 E_J)$ is the Josephson inductance. In this approximation, the transmon behaves as an $LC$ oscillator with plasma frequency $\omega_p = 1/\sqrt{L_J C_J}$.

For harmonic time dependence $\Phi(x,t) = \phi(x)e^{-i\omega t}$, the wave equation reduces to
\begin{equation}
-\phi''(x) = \frac{\omega^2}{v^2}\phi(x) \equiv \lambda\phi(x),
\label{eq:SL_bulk}
\end{equation}
where $\lambda = \omega^2/v^2$ is the eigenparameter with dimensions of inverse length squared. The boundary condition at $x = L$ becomes
\begin{equation}
\frac{\phi'(L)}{\phi(L)} = \beta\lambda - \gamma,
\label{eq:bc_affine}
\end{equation}
where $\beta = C_J/c$ is an effective length (the length of transmission line whose capacitance equals $C_J$), and $\gamma = \ell/L_J$ has dimensions of inverse length. The dimensionless ratio $\beta/L = C_J/(cL)$ compares the junction capacitance to the total line capacitance. This boundary condition is affine in $\lambda$, placing the problem in the class studied by Fulton and Walter~\cite{Fulton1977, Walter1973}.

\subsection{The Quantum Transmon and Linear Response}

The transmon Hamiltonian is $\hat{H}_J = 4E_C\hat{n}^2 - E_J\cos\hat{\varphi}$, where $E_C = e^2/(2C_J)$ is the charging energy. In the transmon regime $E_J/E_C \gg 1$, the eigenstates can be labeled $|g\rangle$, $|e\rangle$, $|f\rangle$, etc. The transition frequencies are $\omega_q \equiv \omega_{ge} \approx \omega_p - E_C/\hbar$ and $\omega_{ef} = \omega_q + \alpha$, where $\alpha = -E_C/\hbar < 0$ is the anharmonicity.

The quantum linear response of the transmon to an oscillating flux produces a state-dependent susceptance $B_q^{(n)}(\omega)$, defined through $Y_q^{(n)}(\omega) = iB_q^{(n)}(\omega)$ for the lossless case. The susceptance has a geometric capacitive contribution $\omega C_J$ plus a dynamical contribution with poles at the transition frequencies $\pm\omega_{mn}$. Standard circuit QED analysis~\cite{Blais2021, Koch2007} gives the vacuum Rabi coupling
\begin{equation}
g = |Q_{ge}|\sqrt{\frac{\omega_r}{2\hbar(cL)}},
\label{eq:g_definition}
\end{equation}
where $Q_{ge} = \langle g|\hat{Q}|e\rangle$ is the charge matrix element and $cL$ is the total resonator capacitance.

\subsection{The State-Dependent Boundary Function}
\label{sec:boundary_function}

Current conservation at $x = L$ relates the transmission line current to the transmon response. For harmonic time dependence at frequency $\omega$, the current into the transmon is $I = Y_q^{(n)}(\omega) V(L)$, where $Y_q^{(n)}(\omega) = iB_q^{(n)}(\omega)$ is the state-dependent admittance and $V(L) = -i\omega\phi(L)$ is the voltage amplitude. The transmission line delivers current $I = -\phi'(L)/\ell$. Equating these:
\begin{equation}
-\frac{1}{\ell}\phi'(L) = iB_q^{(n)}(\omega) \cdot (-i\omega\phi(L)) = \omega B_q^{(n)}(\omega)\phi(L).
\end{equation}
Rearranging gives the boundary condition
\begin{equation}
\frac{\phi'(L)}{\phi(L)} = -\ell\omega B_q^{(n)}(\omega).
\label{eq:bc_from_admittance}
\end{equation}

The susceptance $B_q^{(n)}(\omega)$ consists of a geometric capacitive term plus a dynamical response from virtual transitions. In the \emph{linear response approximation}, the dynamical part is computed from the Kubo formula for small perturbations around the equilibrium state $|n\rangle$. This approximation is valid when the flux amplitude at the boundary satisfies $|2\pi\Phi(L)/\Phi_0| \ll 1$, or equivalently when the resonator photon number is much smaller than the critical photon number $n_{\text{crit}} = \Delta^2/(4g^2)$.

The linear response susceptance has poles at the transition frequencies $\omega = \pm\omega_{mn}$ for all $m \neq n$. Near a pole, the response is dominated by the resonant transition. The full frequency dependence takes the form
\begin{equation}
B_q^{(n)}(\omega) = \omega C_J + \sum_{m \neq n}\frac{A_{nm}\,\omega}{\omega_{mn}^2 - \omega^2},
\label{eq:susceptance_full}
\end{equation}
where $A_{nm}$ depends on the charge matrix elements and transition frequencies. Substituting into Eq.~\eqref{eq:bc_from_admittance} and converting to the eigenparameter $\lambda = \omega^2/v^2$:
\begin{equation}
\frac{\phi'(L)}{\phi(L)} = -\ell v^2\lambda C_J - \sum_{m \neq n}\frac{\ell v^2 A_{nm}\lambda}{\omega_{mn}^2 - v^2\lambda}.
\label{eq:bc_full_lambda}
\end{equation}

To obtain a rational boundary function suitable for Sturm-Liouville analysis, we make a second approximation: the \emph{pole-dominated approximation}. In the resonant denominators $(\omega_{mn}^2 - \omega^2)$, we retain the exact $\lambda$ dependence since this controls the pole structure. However, in the numerator factors where $\lambda$ appears as a slowly-varying prefactor, we replace $\lambda \approx \lambda_r$ where $\lambda_r = \omega_r^2/v^2$ is a reference value near the resonator frequency. This approximation is controlled by the small parameter $|\omega - \omega_r|/\omega_r \ll 1$ and is valid throughout the dispersive regime where we seek eigenfrequencies.

Under this approximation, the boundary function becomes
\begin{equation}
F^{(n)}(\lambda) \equiv \frac{\phi'(L)}{\phi(L)} = -\beta\lambda + \sum_{m \neq n}\frac{\delta_{nm}}{\lambda - \lambda_{nm}},
\label{eq:F_rational}
\end{equation}
where $\beta = C_J/c$ is an effective length, $\lambda_{nm} = \omega_{nm}^2/v^2$ are the pole locations, and the residues $\delta_{nm}$ absorb the prefactors evaluated at the reference frequency. This rational form---affine in $\lambda$ plus simple poles---places the problem in the class of Sturm-Liouville problems with eigenparameter-dependent boundary conditions studied by Fulton and Walter~\cite{Fulton1977, Walter1973}.

To summarize, two distinct approximations lead to the rational boundary function~\eqref{eq:F_rational}:
\begin{enumerate}
\item \textbf{Linear response:} The transmon response is computed to first order in the perturbing flux, valid for $\bar{n} \ll n_{\text{crit}}$.
\item \textbf{Pole-dominated expansion:} Slowly-varying numerator factors are evaluated at a reference frequency $\omega_r$, valid for $|\omega - \omega_r| \ll \omega_r$.
\end{enumerate}
The first approximation is the standard linearization underlying all dispersive readout theory. The second is an additional simplification that renders the boundary function rational in $\lambda$, enabling the application of Fulton-Walter spectral theory. Both approximations are well-satisfied in typical circuit QED experiments operating in the dispersive regime. The residues $\delta_{nm}$ have dimensions of inverse length cubed (since $F$ has dimensions $[\text{L}^{-1}]$ and the pole terms have $[\delta]/[\lambda] = [\delta]/[\text{L}^{-2}]$). Their values are not immediately apparent from the susceptance expression~\eqref{eq:susceptance_full} due to the approximations involved. Instead, we determine them by matching to a known physical result: the vacuum Rabi splitting at resonance. This matching procedure, carried out in Section~\ref{sec:spectral}, yields
\begin{equation}
\delta_{nm} = \frac{2L g_{nm}^2 \omega_{nm}^2}{v^4},
\label{eq:residue_from_matching}
\end{equation}
where $g_{nm}$ is the vacuum Rabi coupling for the $n \leftrightarrow m$ transition.

The sign of $\delta_{nm}$ equals the sign of $\omega_{mn} = (E_m - E_n)/\hbar$: positive for absorption ($E_m > E_n$) and negative for emission ($E_m < E_n$). For the ground state $|g\rangle$, all accessible transitions are absorptive, so all residues are positive. For the excited state $|e\rangle$, the residue $\delta_{eg} = -\delta_{ge} < 0$ corresponds to stimulated emission to the ground state, while $\delta_{ef} > 0$ corresponds to absorption to the second excited state. The harmonic oscillator matrix element scaling $|Q_{ef}| \approx \sqrt{2}|Q_{ge}|$, valid for weakly anharmonic transmons, implies $|\delta_{ef}| \approx 2|\delta_{ge}|$.

\subsection{The Extended Hilbert Space}

Sturm-Liouville problems with $M$ poles in the boundary function require an extended Hilbert space
\begin{equation}
\mathcal{H}_{\text{ext}} = L^2(0,L) \oplus \mathbb{C}^M
\label{eq:extended_hilbert}
\end{equation}
to accommodate the boundary dynamics~\cite{Fulton1977, Guliyev2025}. The $L^2(0,L)$ factor describes the flux field profile, while $\mathbb{C}^M$ represents boundary amplitudes $\xi_k$ associated with each pole, satisfying $\xi_k = \sqrt{|\delta_k|}\phi(L)/(\lambda - \lambda_k)$. This extended space is a mathematical device for solving the eigenvalue problem with a fixed transmon state $|n\rangle$. It should be distinguished from the physical Hilbert space $\mathcal{H}_{\text{phys}} = \mathcal{H}_{\text{res}} \otimes \mathcal{H}_{\text{transmon}}$ of the coupled quantum system. The boundary amplitudes $\xi_k$ represent virtual excitation amplitudes of transitions, not probability amplitudes for transmon states. Different transmon eigenstates define different boundary functions, hence different eigenvalue problems, which is what enables dispersive readout.

The extended Hilbert space $\mathcal{H}_{\mathrm{ext}}^{(n)} = L^2(0,L) \oplus \mathbb{C}^M$ invites comparison with the restricted Hilbert space of the wedge geometry, as both represent mathematical encodings of physical constraints. In the wedge, Dirichlet boundary conditions impose a kinematic constraint: the physical Hilbert space is $L^2([0,\Phi])$ 
with functions vanishing at the boundaries, and the domain restriction to 
sine modes is the natural encoding of this geometric fact. In dispersive 
readout, the transmon imposes a dynamical constraint: the boundary 
admittance depends on the eigenvalue $\lambda = \omega^2/v^2$, and the 
Hilbert space extension to $L^2 \oplus \mathbb{C}^M$ is the natural encoding of this frequency-dependent response within Sturm-Liouville theory. Neither formulation is more fundamental; each is the appropriate mathematical framework for its physical context. The $\mathbb{C}^M$ components are not additional physical degrees of freedom but auxiliary variables that capture the boundary's spectral structure at each pole of the admittance function. Despite these procedural differences, the functional role is identical: given a definite constraint (wall position or transmon state), the boundary condition selects which modes are compatible with the system.

\subsection{Derivation of the Dispersive Shift}

We adopt the standard convention $\Delta \equiv \omega_q - \omega_r$ throughout, with $\Delta < 0$ for a qubit below the resonator. The dressed mode frequencies satisfy $G(\lambda) = F^{(n)}(\lambda)$, where $G(\lambda) = \sqrt{\lambda}\cot(\sqrt{\lambda}L)$. Near a bare resonator mode at $\lambda_r = \omega_r^2/v^2$, perturbation theory gives the frequency shift
\begin{equation}
\delta\omega^{(n)} = -\frac{v^2}{\omega_r L} F^{(n)}(\lambda_r).
\label{eq:freq_shift}
\end{equation}

For the ground state with a single dominant pole at $\lambda_q$:
\begin{equation}
\delta\omega^{(g)} = -\frac{v^2}{\omega_r L} \cdot \frac{\delta_{ge}}{\lambda_r - \lambda_q}.
\end{equation}
Converting to frequency using $\lambda_r - \lambda_q = (\omega_r^2 - \omega_q^2)/v^2 \approx -2\omega_r\Delta/v^2$:
\begin{equation}
\delta\omega^{(g)} = \frac{v^4 \delta_{ge}}{2\omega_r^2 L \Delta}.
\end{equation}

For the excited state, including both $e \to g$ (emission, $\delta_{eg} = -\delta_{ge}$) and $e \to f$ (absorption, $\delta_{ef} \approx 2\delta_{ge}$):
\begin{equation}
\delta\omega^{(e)} = \frac{v^4 \delta_{ge}}{2\omega_r^2 L}\left[-\frac{1}{\Delta} + \frac{2}{\Delta + \alpha}\right].
\end{equation}

The dispersive shift is half the frequency difference:
\begin{align}
2\chi &= \delta\omega^{(e)} - \delta\omega^{(g)} \nonumber\\
&= \frac{v^4 \delta_{ge}}{2\omega_r^2 L}\left[-\frac{2}{\Delta} + \frac{2}{\Delta + \alpha}\right] \nonumber\\
&= \frac{v^4 \delta_{ge}}{\omega_r^2 L} \cdot \frac{-\alpha}{\Delta(\Delta + \alpha)}.
\end{align}

Substituting $\delta_{ge} = 2Lg^2\omega_q^2/v^4$ from Eq.~\eqref{eq:residue_from_matching} and using $\omega_q \approx \omega_r$ in the dispersive regime:
\begin{equation}
\chi = \frac{g^2\alpha}{\Delta(\Delta + \alpha)}.
\label{eq:chi_final}
\end{equation}
This is the standard dispersive shift formula~\cite{Koch2007, Blais2021}. For $\alpha < 0$ and $\Delta < 0$, we have $\chi < 0$: the resonator frequency is lower when the qubit is excited.

\subsection{Approximations and Regime of Validity}

The derivation relied on several approximations. The linearization of the Josephson nonlinearity requires $\bar{n} \ll n_{\text{crit}} = \Delta^2/(4g^2)$. The dispersive regime requires $|g| \ll |\Delta|$. The transmon level truncation is valid when higher transitions are sufficiently detuned; for typical parameters with $|\alpha|/2\pi \sim 200$~MHz and $|\Delta|/2\pi \sim 500$~MHz, the three-level model captures the dominant physics. The single-mode approximation neglects higher resonator modes, which contribute to renormalization effects discussed in Section~\ref{sec:spectral}.


\section{Spectral Structure and Multimode Physics}
\label{sec:spectral}

The boundary condition derived in the previous section leads naturally to a transcendental eigenvalue problem whose solutions determine the dressed mode frequencies of the coupled system. This section analyzes the spectral structure, proves a level repulsion theorem, discusses the vacuum Rabi splitting at resonance, and addresses the multimode physics that becomes important when all resonator modes are included.

\subsection{The Eigenvalue Equation}

The dressed mode frequencies satisfy the transcendental equation
\begin{equation}
G(\lambda) = F^{(n)}(\lambda),
\label{eq:eigenvalue_equation}
\end{equation}
where $\lambda = \omega^2/v^2$ is the squared wavenumber, $v$ is the phase velocity, and $L$ is the resonator length. The function
\begin{equation}
G(\lambda) = \sqrt{\lambda}\cot(\sqrt{\lambda}\,L)
\label{eq:G_def}
\end{equation}
encodes the bare resonator boundary condition at the grounded end, while $F^{(n)}(\lambda)$ encodes the transmon's state-dependent response at the coupling end. This equation arises from requiring nontrivial solutions to the wave equation $\partial_x^2\phi = -\lambda\phi$ subject to $\phi(0) = 0$ at the grounded end and the frequency-dependent boundary condition $\phi'(L)/\phi(L) = F^{(n)}(\lambda)$ at the transmon. Figure~\ref{fig:eigenvalue_graphical} illustrates the graphical solution: the dressed eigenvalues occur at intersections of $G(\lambda)$ and $F^{(n)}(\lambda)$, with the monotonicity properties of each function ensuring exactly one solution between consecutive poles.

\begin{figure}[t]
\centering
\includegraphics[width=\columnwidth]{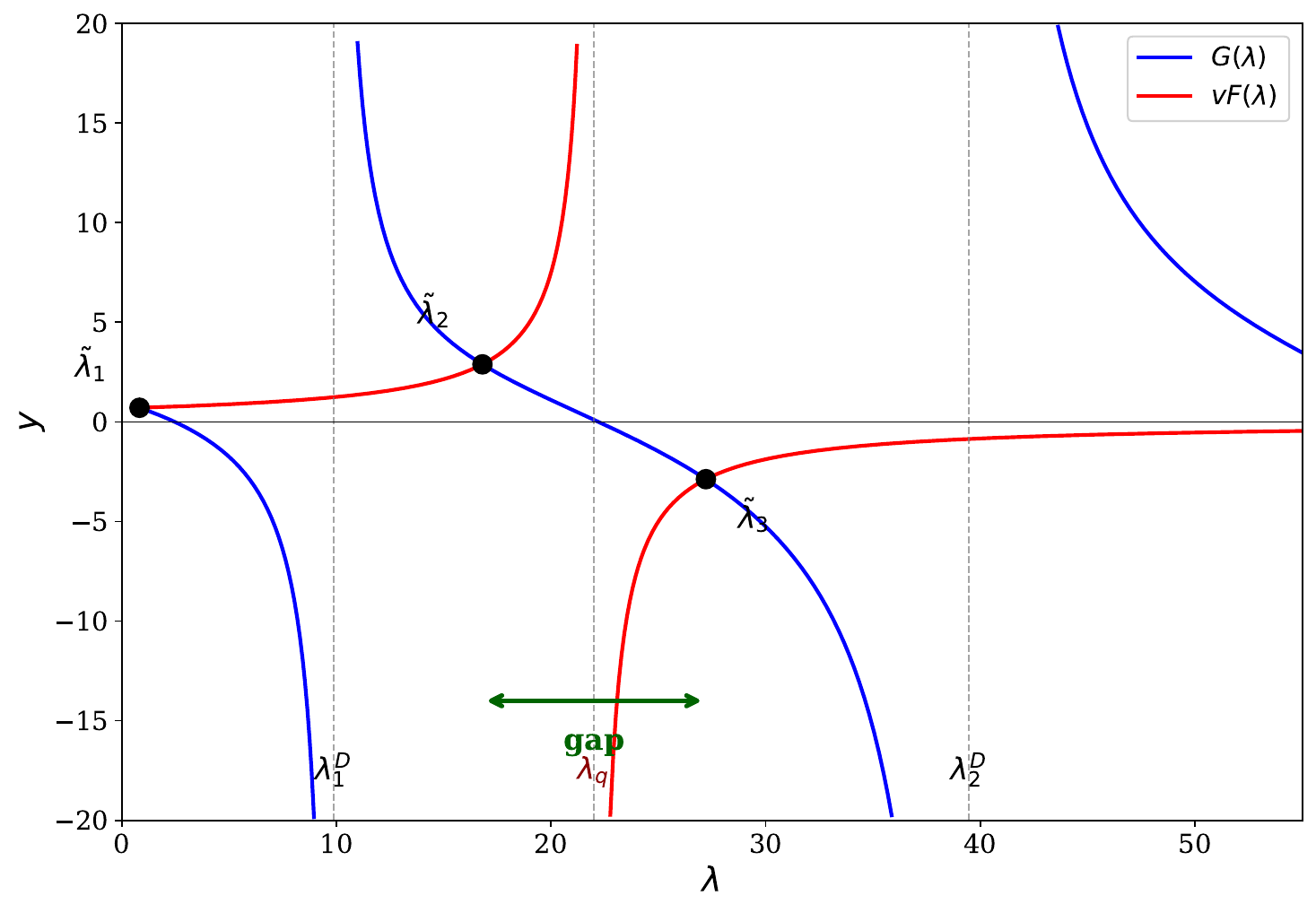}
\caption{Graphical solution of the eigenvalue equation $G(\lambda) = F^{(g)}(\lambda)$ 
for a transmon in the ground state. The function $G(\lambda)$ (blue) has poles at Dirichlet eigenvalues $\lambda_n^D$ and decreases monotonically between them. The boundary function $F^{(g)}(\lambda)$ (red) has a pole at the qubit transition $\lambda_q$ with positive residue. Dressed eigenvalues $\tilde{\lambda}_1, \tilde{\lambda}_2, \tilde{\lambda}_3$ (black dots) occur at intersections. For the ground state, where all residues are positive, exactly one eigenvalue lies between each consecutive pair of poles. The level repulsion theorem guarantees that no eigenvalue coincides with $\lambda_q$; near resonance, the two neighboring eigenvalues are separated by a gap corresponding to the vacuum Rabi splitting.}
\label{fig:eigenvalue_graphical}
\end{figure}

\subsection{Properties of the Resonator Function $G(\lambda)$}

The function $G(\lambda)$ has poles where $\sin(\sqrt{\lambda}\,L) = 0$, occurring at the Dirichlet eigenvalues
\begin{equation}
\lambda_k^D = \left(\frac{k\pi}{L}\right)^2, \quad k = 1, 2, 3, \ldots
\end{equation}
Between consecutive Dirichlet poles, $G(\lambda)$ has exactly one zero at the bare quarter-wave frequencies $\lambda_k^{(0)} = ((2k-1)\pi/(2L))^2$.

At each pole $\lambda_k^D$, the function $G(\lambda) \to +\infty$ as $\lambda \to (\lambda_k^D)^+$ and $G(\lambda) \to -\infty$ as $\lambda \to (\lambda_k^D)^-$. Between consecutive poles, $G(\lambda)$ is strictly decreasing, which can be verified by showing that $d(\xi\cot\xi)/d\xi < 0$ for $\xi > 0$ away from poles.

The derivative at a bare quarter-wave frequency is
\begin{equation}
G'(\lambda_k^{(0)}) = -\frac{L}{2},
\label{eq:G_prime}
\end{equation}
a result that follows from evaluating $(d/d\xi)(\xi\cot\xi)$ at $\xi = (2k-1)\pi/2$ where $\cot\xi = 0$ and $\csc^2\xi = 1$.

\subsection{Properties of the Boundary Function}

The boundary function for a transmon in state $|n\rangle$ has poles at the transition frequencies $\lambda_{nm} = \omega_{nm}^2/v^2$ for all $m \neq n$. The residue at each pole has the same sign as $\omega_{nm}$: positive for absorption ($E_m > E_n$) and negative for emission ($E_m < E_n$).

For the ground state $|g\rangle$, all accessible transitions are absorptive, so all residues are positive. For excited states, both positive and negative residues appear, corresponding to absorption to higher levels and emission to lower levels respectively.

\subsection{Level Repulsion}

The pole structure of both $G(\lambda)$ and $F^{(n)}(\lambda)$ leads to a fundamental constraint on the spectrum.

\begin{theorem}[Level Repulsion]
\label{thm:level_repulsion}
No eigenvalue of the coupled system can coincide with a transmon transition frequency. That is, if $\lambda^*$ satisfies $G(\lambda^*) = F^{(n)}(\lambda^*)$, then $\lambda^* \neq \lambda_{nm}$ for any $m \neq n$.
\end{theorem}

\begin{proof}
At $\lambda = \lambda_{nm}$, the boundary function $F^{(n)}(\lambda)$ diverges due to its pole, while $G(\lambda_{nm})$ remains finite under the generic assumption that no transmon transition frequency coincides with a Dirichlet eigenvalue of the bare resonator. Therefore the equation $G(\lambda_{nm}) = F^{(n)}(\lambda_{nm})$ cannot be satisfied.
\end{proof}

This level repulsion implies that as a bare resonator mode is tuned toward a transmon transition frequency, the dressed eigenvalue approaches the transition frequency asymptotically but cannot cross it. The eigenvalue must curve away, producing the characteristic avoided crossing observed in spectroscopy.

For the ground state where all residues are positive, a stronger result holds: between any two consecutive poles in the combined set of Dirichlet eigenvalues and transmon transitions, exactly one dressed eigenvalue exists. This interlacing property provides a global constraint on the spectrum. For excited states with mixed-sign residues, the monotonicity properties that guarantee interlacing may fail, and a more detailed analysis is required to determine the number of eigenvalues in each interval.

\subsection{Vacuum Rabi Splitting}

When the qubit frequency is tuned to resonance with a bare resonator mode, $\omega_q = \omega_r$, the level repulsion theorem guarantees that no dressed eigenvalue equals the common frequency. Instead, two dressed modes appear symmetrically displaced from resonance. We now compute this vacuum Rabi splitting and use it to determine the residue formula.

Consider the single-pole approximation where only the $g \leftrightarrow e$ transition is retained, with the qubit tuned to resonance with the fundamental mode: $\lambda_q = \lambda_r^{(0)}$. At this point, $G(\lambda_r^{(0)}) = 0$ (the bare open-circuit condition) and the boundary function reduces to
\begin{equation}
F(\lambda) \approx -\beta\lambda + \frac{\delta}{\lambda - \lambda_q},
\end{equation}
where we write $\delta \equiv \delta_{ge}$ for brevity.

Near $\lambda = \lambda_q$, we expand $G(\lambda)$ to first order:
\begin{equation}
G(\lambda) \approx G'(\lambda_q)(\lambda - \lambda_q) = -\frac{L}{2}(\lambda - \lambda_q),
\end{equation}
using $G'(\lambda_r^{(0)}) = -L/2$ from Eq.~\eqref{eq:G_prime}. The eigenvalue equation $G(\lambda) = F(\lambda)$ becomes
\begin{equation}
-\frac{L}{2}(\lambda - \lambda_q) = -\beta\lambda_q + \frac{\delta}{\lambda - \lambda_q},
\label{eq:eigenvalue_near_resonance}
\end{equation}
where we have used $-\beta\lambda \approx -\beta\lambda_q$ for $\lambda$ near $\lambda_q$.

Let $\epsilon = \lambda - \lambda_q$. Equation~\eqref{eq:eigenvalue_near_resonance} becomes
\begin{equation}
-\frac{L}{2}\epsilon = -\beta\lambda_q + \frac{\delta}{\epsilon}.
\end{equation}
Multiplying by $\epsilon$ and rearranging:
\begin{equation}
\frac{L}{2}\epsilon^2 - \beta\lambda_q\epsilon - \delta = 0.
\end{equation}
The solutions are
\begin{equation}
\epsilon_\pm = \frac{\beta\lambda_q \pm \sqrt{\beta^2\lambda_q^2 + 2L\delta}}{L}.
\end{equation}
For weak coupling where $2L\delta \ll \beta^2\lambda_q^2$, the two roots are far apart: one near $2\beta\lambda_q/L$ (a large shift) and one near $-\delta/(\beta\lambda_q)$ (a small shift). This limit does not produce symmetric splitting.

The symmetric splitting characteristic of vacuum Rabi oscillations occurs when the pole term dominates over the affine term near resonance. In this regime, $|\delta/\epsilon| \gg \beta\lambda_q$, and the eigenvalue equation simplifies to
\begin{equation}
-\frac{L}{2}\epsilon \approx \frac{\delta}{\epsilon},
\end{equation}
giving
\begin{equation}
\epsilon^2 = -\frac{2\delta}{L}.
\end{equation}
For $\delta > 0$ (ground state, absorption), this equation has no real solutions, the dressed eigenvalues do not cross $\lambda_q$. The physical solutions lie on either side of the pole, found by analyzing the full equation in the intervals $(\lambda_{k-1}^D, \lambda_q)$ and $(\lambda_q, \lambda_k^D)$ where continuity arguments guarantee one root each.

To find the splitting, we note that near the pole, the dominant balance gives $|\epsilon| \sim \sqrt{2\delta/L}$. Converting to frequency using $\lambda = \omega^2/v^2$ and $\epsilon = (\omega^2 - \omega_q^2)/v^2 \approx 2\omega_q\delta\omega/v^2$ for small frequency shifts:
\begin{equation}
\delta\omega \approx \pm\frac{v^2}{2\omega_q}\sqrt{\frac{2\delta}{L}}.
\end{equation}
The total splitting between the two dressed modes is
\begin{equation}
\Delta\omega = 2|\delta\omega| = \frac{v^2}{\omega_q}\sqrt{\frac{2\delta}{L}}.
\label{eq:splitting_from_delta}
\end{equation}

The Jaynes-Cummings model predicts a vacuum Rabi splitting of $\Delta\omega = 2g$ at resonance. Equating this to Eq.~\eqref{eq:splitting_from_delta}:
\begin{equation}
2g = \frac{v^2}{\omega_q}\sqrt{\frac{2\delta}{L}}.
\end{equation}
Solving for the residue:
\begin{equation}
\delta = \frac{2Lg^2\omega_q^2}{v^4}.
\label{eq:residue_derived}
\end{equation}
This is the residue formula quoted in Section~\ref{sec:boundary_function}. At resonance $\omega_q = \omega_r$, it can be written as
\begin{equation}
\delta_{ge} = \frac{2Lg^2\omega_r^2}{v^4} = \frac{\omega_r^3|Q_{ge}|^2}{2\hbar v^3 Z_0},
\label{eq:residue_explicit}
\end{equation}
where the second equality uses $g^2 = \omega_r|Q_{ge}|^2/(2\hbar cL)$ and $cL = L/(vZ_0)$.

The derivation confirms that the Sturm-Liouville formulation reproduces the Jaynes-Cummings vacuum Rabi splitting when the residue is given by Eq.~\eqref{eq:residue_derived}. This matching procedure determines the otherwise-unknown prefactors that arose from the pole-dominated approximation in Section~\ref{sec:boundary_function}. For transitions other than $g \leftrightarrow e$, the same analysis applies with the appropriate coupling $g_{nm}$ and transition frequency $\omega_{nm}$:
\begin{equation}
\delta_{nm} = \frac{2Lg_{nm}^2\omega_{nm}^2}{v^4}.
\label{eq:residue_general}
\end{equation}
The sign convention $\delta_{nm} > 0$ for $\omega_{nm} > 0$ (absorption) and $\delta_{nm} < 0$ for $\omega_{nm} < 0$ (emission) follows from the structure of the linear response susceptance.

\subsection{Relation to the Jaynes-Cummings Model}

\begin{figure}[t]
\centering
\includegraphics[width=\columnwidth]{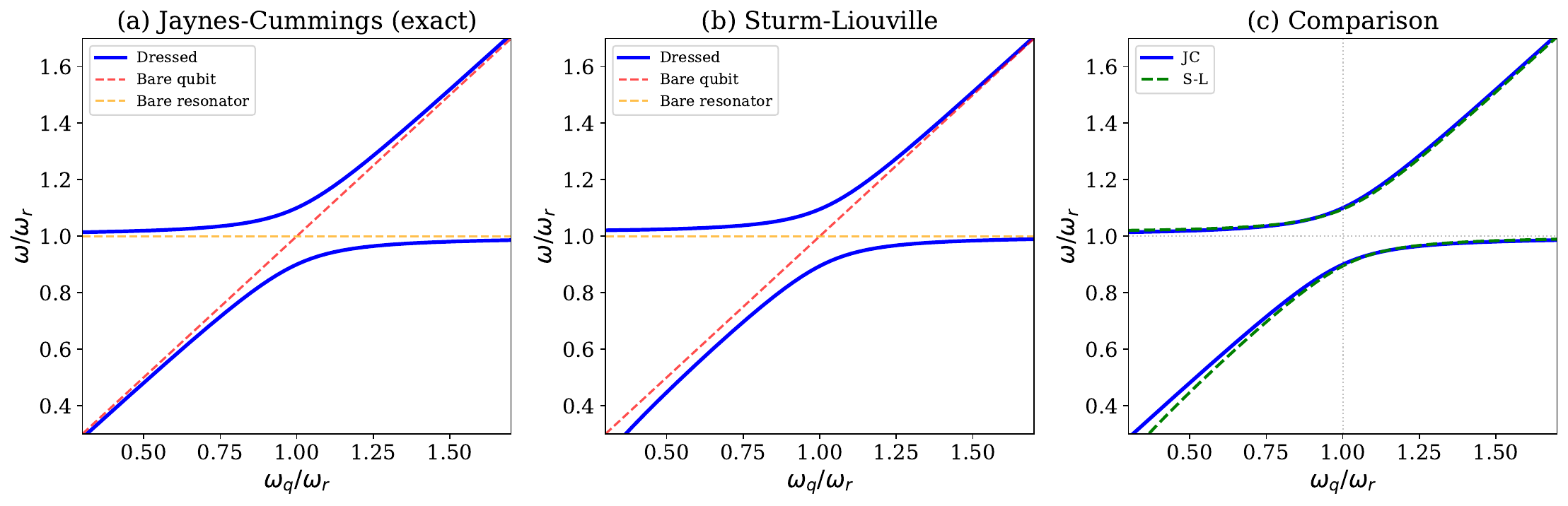}
\caption{Comparison of dressed mode frequencies from (a) the Jaynes-Cummings model and (b) the Sturm-Liouville formulation as the qubit frequency $\omega_q$ is tuned through resonance with the bare resonator mode $\omega_r$. Both exhibit the characteristic avoided crossing with gap $2g$ at resonance. Panel (c) overlays the two calculations, showing close agreement when the residue is matched to the coupling via Eq.~\eqref{eq:residue_derived}. The coupling strength is $g/\omega_r = 0.15$.}
\label{fig:SL_JC_comparison}
\end{figure}

The Jaynes-Cummings Hamiltonian
\begin{equation}
H_{\text{JC}} = \hbar\omega_r a^\dagger a + \frac{\hbar\omega_q}{2}\sigma_z + \hbar g(a^\dagger\sigma^- + a\sigma^+)
\end{equation}
employs the rotating-wave approximation (RWA), which neglects the counter-rotating terms $a\sigma^-$ and $a^\dagger\sigma^+$. These terms contribute corrections of order $g^2/(\omega_q + \omega_r)$, known as Bloch-Siegert shifts. Figure~\ref{fig:SL_JC_comparison} compares the dressed mode frequencies obtained from the Jaynes-Cummings model and the Sturm-Liouville formulation as the qubit frequency is tuned through resonance; the two approaches yield closely matching avoided crossings when the residue is related to the coupling strength via Eq.~\eqref{eq:residue_derived}.

The boundary-condition formulation developed here does not make the RWA at the level of the wave equation. However, the dispersive-regime approximation used in deriving the rational boundary function $F(\lambda)$ introduces simplifications that may affect the treatment of counter-rotating contributions. A complete analysis of Bloch-Siegert corrections within this framework would require retaining the full frequency dependence of the transmon susceptibility without the dispersive approximation, which is beyond the scope of this paper. For typical circuit QED parameters where $g/\omega_r \sim 0.01$--$0.1$, the Bloch-Siegert corrections are small and the Jaynes-Cummings model provides an excellent approximation.

\subsection{Multimode Physics and UV Divergences}

The transmission line supports infinitely many modes, each coupled to the transmon. For large mode number $n$, the mode frequency scales as $\omega_n \propto n$ and the coupling strength scales as
\begin{equation}
g_n \propto \sqrt{\omega_n} \propto \sqrt{n}.
\end{equation}
This scaling follows from the normalization of the mode functions and the capacitive nature of the coupling.

The Lamb shift---the total frequency renormalization of the qubit due to coupling to all modes---is given to second order by
\begin{equation}
\delta\omega_q^{\text{Lamb}} = \sum_{n=1}^{\infty}\frac{g_n^2}{\omega_q - \omega_n}.
\end{equation}
For large $n$, the summand approaches a constant: $g_n^2/(\omega_q - \omega_n) \approx g_n^2/(-\omega_n) \propto n/(-n) = -1$. The Lamb shift therefore diverges linearly with the mode cutoff:
\begin{equation}
\delta\omega_q^{\text{Lamb}} \sim -N_{\text{max}} \to -\infty.
\end{equation}

This ultraviolet divergence signals that the ``bare'' qubit frequency in the Hamiltonian differs from the physical (measured) frequency. Renormalization absorbs this divergence by adjusting the bare frequency so that the physical frequency remains finite. This procedure is standard in quantum electrodynamics and its circuit analog.

Physical mechanisms that provide natural UV cutoffs include the junction plasma frequency, the superconducting gap, and dispersion from kinetic inductance. A complete multimode theory would incorporate these effects to obtain finite predictions without artificial cutoffs.

The dispersive shift, unlike the Lamb shift, is much less divergent. The contribution from mode $n$ scales as
\begin{equation}
\chi_n \propto \frac{g_n^2\alpha}{\Delta_n^2} \propto \frac{n}{n^2} = \frac{1}{n}
\end{equation}
for far-detuned modes, giving a sum that diverges only logarithmically. In practice, the dispersive approximation breaks down for very far-detuned modes, further suppressing their contribution.

\subsection{The Extended Hilbert Space}

The Fulton-Walter extended space $\mathcal{H}_{\text{ext}}^{(n)} = L^2(0,L) \oplus \mathbb{C}^{M}$ was introduced in Section~\ref{sec:dynamic_bc}. The key physical insight bears repeating: different transmon states define different boundary conditions, hence different resonator mode frequencies, enabling dispersive readout. The Sturm-Liouville framework complements Hamiltonian diagonalization, with the level repulsion theorem providing a foundation for the avoided crossings central to circuit QED physics.

\section{Multi-Qubit Extension: Parity-Preserving Readout}
\label{sec:parity}

The single-qubit analysis showed how the dispersive shift emerges from a state-dependent boundary condition. We now extend this framework to multiple qubits coupled to a shared resonator, showing how the combined boundary admittance determines the joint-state-dependent resonator frequencies and under what conditions parity becomes a quantum non-demolition observable.

As in the single-qubit case, we work within the single-mode approximation and use the quantum linear response of each transmon as input. The extension to multiple qubits follows naturally from the principle that admittances add in parallel at a common boundary.

Consider two transmon qubits coupled to a shared transmission line resonator. Each qubit contributes an effective admittance at the boundary, and Kirchhoff's current law requires that the current from the transmission line equals the sum of currents into both qubits. Within the linear-response approach, the total effective admittance at the boundary is the sum of individual contributions:
\begin{equation}
Y_{\text{total}}^{(n_1, n_2)}(\omega) = Y_{\text{res}}(\omega) + Y_{q,1}^{\text{eff},(n_1)}(\omega) + Y_{q,2}^{\text{eff},(n_2)}(\omega),
\label{eq:Y_total_two_qubit}
\end{equation}
where we use the state-dependent admittance $Y_q^{(n)}(\omega) = iB_q^{(n)}(\omega)$ for each qubit. The dressed resonator frequencies for each joint state $(n_1, n_2)$ are determined by $Y_{\text{total}}^{(n_1, n_2)}(\omega) = 0$.

Because admittances add linearly, the frequency shifts also add:
\begin{equation}
\delta\omega_r^{(n_1, n_2)} = \delta\omega_r^{(1, n_1)} + \delta\omega_r^{(2, n_2)},
\label{eq:shift_additivity}
\end{equation}
where $\delta\omega_r^{(j, n_j)}$ is the shift due to qubit $j$ in state $n_j$. Using the dispersive shift $\chi_j = g_j^2\alpha_j/[\Delta_j(\Delta_j + \alpha_j)]$ for each qubit, and absorbing state-independent Lamb shifts into a renormalized $\omega_r$, the four computational basis state frequencies become:
\begin{align}
\omega_r(|gg\rangle) &= \omega_r + \chi_1 + \chi_2, \label{eq:freq_gg}\\
\omega_r(|ge\rangle) &= \omega_r + \chi_1 - \chi_2, \label{eq:freq_ge}\\
\omega_r(|eg\rangle) &= \omega_r - \chi_1 + \chi_2, \label{eq:freq_eg}\\
\omega_r(|ee\rangle) &= \omega_r - \chi_1 - \chi_2. \label{eq:freq_ee}
\end{align}
These can be written compactly as $\omega_r(|n_1, n_2\rangle) = \omega_r + \chi_1\sigma_1^z + \chi_2\sigma_2^z$, corresponding to the dispersive Hamiltonian
\begin{equation}
H_{\text{disp}} = \hbar\omega_r \hat{n} + \hbar\chi_1\sigma_1^z\hat{n} + \hbar\chi_2\sigma_2^z\hat{n}.
\label{eq:H_disp}
\end{equation}

The joint parity operator $P = \sigma_1^z\sigma_2^z$ has eigenvalue $+1$ for even-parity states ($|gg\rangle$, $|ee\rangle$) and $-1$ for odd-parity states ($|ge\rangle$, $|eg\rangle$). The dispersive Hamiltonian commutes with parity: $[H_{\text{disp}}, P] = 0$. This holds because each term either involves $\hat{n}$ alone or is proportional to $\sigma_j^z$, both of which commute with $P = \sigma_1^z\sigma_2^z$. In practice, this QND property holds to the extent the dynamics is well-described by the dispersive Hamiltonian and leakage to non-computational states is negligible; measurement-induced transitions can occur due to non-dispersive terms, strong drives, or transmon leakage.

When dispersive shifts are matched, $\chi_1 = \chi_2 \equiv \chi$, the frequencies simplify: $\omega_r(|gg\rangle) = \omega_r + 2\chi$, $\omega_r(|ge\rangle) = \omega_r(|eg\rangle) = \omega_r$, and $\omega_r(|ee\rangle) = \omega_r - 2\chi$. The odd-parity states become frequency-degenerate, but the even-parity states remain distinguishable by $4\chi$. In the boundary condition picture, this occurs because the total admittance is symmetric under qubit exchange within the odd-parity subspace.

This corrects a common misconception: $\chi_1 = \chi_2$ does not make the resonator respond only to parity. The Hamiltonian contains the linear term $\chi(\sigma_1^z + \sigma_2^z)\hat{n}$, which responds to total magnetization, not parity. The frequency structure is $\omega_r + \chi(\sigma_1^z + \sigma_2^z)$, where $\sigma_1^z + \sigma_2^z \in \{+2, 0, 0, -2\}$, not $\sigma_1^z\sigma_2^z \in \{+1, -1, -1, +1\}$.

The implications for coherence are significant. The odd-parity coherence $\langle ge|\rho|eg\rangle$ is protected under matched dispersive shifts since no which-path information distinguishes these states. The even-parity coherence $\langle gg|\rho|ee\rangle$ is not protected: these states produce different frequencies, generating which-path information that dephases their superposition.

True parity-only measurement requires all states within each parity sector to produce identical resonator responses, meaning the linear term $\chi(\sigma_1^z + \sigma_2^z)\hat{n}$ must be cancelled while retaining a parity-dependent term. Several approaches achieve this: Royer, Puri, and Blais use parametric driving to cancel linear dispersive shifts~\cite{Royer2018}; DiVincenzo and Solgun analyze circuit geometries with inherent parity sensitivity~\cite{DiVincenzo2013}; Lalumi\`ere, Gambetta, and Blais study tunable joint measurements~\cite{Lalumiere2010}. In these schemes, the effective Hamiltonian becomes $H_{\text{parity}} = \hbar\omega_r\hat{n} + \hbar\chi_P P\hat{n}$, where $\chi_P$ is an engineered parity-dependent shift and linear terms are absent.

From the boundary condition perspective, these schemes engineer an effective admittance that depends on $P = \sigma_1^z\sigma_2^z$ rather than on $\sigma_1^z$ and $\sigma_2^z$ individually. The individual qubit contributions interfere destructively for linear terms while reinforcing the parity-dependent component.

\section{Structural Correspondence to Quantum Error Correction}
\label{sec:qec}

The parity measurement analyzed in Section~\ref{sec:parity} is a specific instance of stabilizer measurement. We now discuss the broader structural correspondence between stabilizer constraints and boundary conditions, while being explicit about the limits of this correspondence. In stabilizer-based quantum error correction, the code space is defined as the simultaneous $+1$ eigenspace of a set of commuting stabilizer operators $\{S_i\}$:
\begin{equation}
\mathcal{H}_{\text{code}} = \{\ket{\psi} : S_i\ket{\psi} = +\ket{\psi} \text{ for all } i\}.
\end{equation}
This definition has the same structure as a boundary condition: it is a constraint that selects a subspace of compatible states from a larger Hilbert space. For the parity stabilizer $S = Z_1 Z_2$, this correspondence is physical: the dispersive measurement analyzed in Section~\ref{sec:parity} implements a Sturm-Liouville problem where the boundary function $F(\lambda)$ depends on the parity eigenvalue, with parity eigenstates defining distinct boundary conditions.

Errors that anticommute with a stabilizer map states outside the code space: $S_i (E\ket{\psi}) = -E\ket{\psi}$ if $\{S_i, E\} = 0$. This is analogous to how $L_z$ maps Dirichlet-satisfying functions to functions that violate the boundary conditions. Syndrome measurement detects which stabilizers have been violated, and error correction returns the state to the code space.

However, general stabilizer codes involve stabilizers that are not simple $Z$-type parities. Weight-four stabilizers like $X_1 X_2 X_3 X_4$ in the surface code require different physical implementations, typically involving ancilla qubits and sequences of two-qubit gates rather than direct dispersive coupling. Measuring such $X$-type stabilizers requires first rotating the computational basis via Hadamard gates, performing CNOT operations to transfer parity information to an ancilla, and then measuring the ancilla. While this still implements a stabilizer constraint, the connection to boundary conditions becomes structural rather than derived from circuit physics: the stabilizer defines which states are ``allowed,'' but the physical mechanism differs from the direct dispersive coupling analyzed here. The structural insight remains valuable: stabilizer constraints, like boundary conditions, define a subspace of compatible states; errors, like operators that violate boundary conditions, map states outside this subspace; syndrome measurement identifies which constraints have been violated. This parallel provides intuition for why quantum error correction works, even when the physical implementation differs from dispersive parity measurement.

\section{The Born Rule}
\label{sec:born}

Once a measurement interaction has dynamically selected a commuting algebra of effective observables, the remaining question concerns the assignment of probabilities within this algebra. After decoherence, the global state has orthogonal branches, and the reduced density matrix assigns weights $|c_n|^2$ to the eigenstates, where $c_n = \langle n|\psi\rangle$ is the initial amplitude.

The status of these weights as probabilities is a foundational question that we do not attempt to resolve. Under standard noncontextuality assumptions, Gleason-type theorems~\cite{Gleason1957} uniquely determine the probability measure for systems of dimension three or greater: any probability assignment satisfying natural consistency conditions must be given by the Born rule. For two-dimensional systems, analogous conclusions require additional structural assumptions. The representation-theoretic structure explains why eigenstates of the measured observable form the natural basis: they are the irreducible representations of the symmetry preserved by $\Hint$. Decoherence explains why interference between branches is suppressed. What this framework does not explain is why any particular observer finds themselves in one branch rather than another, whether this is a genuine physical question or an indexical one remains contested~\cite{Wallace2012}.

We have chosen to keep this discussion brief because the Born rule lies outside the main scope of this paper, which concerns basis selection rather than outcome probabilities. A fuller treatment would require engaging with the extensive literature on probability in quantum mechanics, the decision-theoretic approaches in Everettian quantum mechanics, and the status of objective probability more generally, topics that deserve dedicated analysis rather than cursory treatment.

\section{Scope and Interpretation}
\label{sec:open}

The framework developed here provides a unified mathematical structure, Sturm-Liouville problems with eigenparameter-dependent boundary conditions, that describes both geometric constraints in confined systems and measurement interactions in circuit QED. This formalism yields concrete results: the dispersive shift formula, interlacing theorems for dressed modes, and the structure of parity-preserving measurements.

The framework addresses basis selection: why the measurement 
Hamiltonian determines which observable is measured, and why its 
eigenstates become operationally stable under environmental decoherence. The transmon's self-Hamiltonian creates state-dependent poles in the boundary admittance; these poles, not the specifics of decoherence, determine which observables commute with the interaction and thus remain well-defined.

The framework takes environmental decoherence as an established 
mechanism and focuses on its structural consequences for observable 
selection. It does not address outcome selection: why any 
particular observer reports a specific result. In the wedge, Dirichlet 
conditions eliminate modes from the Hilbert space; in measurement, 
decoherence correlates branches without eliminating them. The question 
"which branch is actual?" lies outside this analysis and remains 
contested in quantum foundations.

The geometric perspective is compatible with multiple interpretations: 
Everettian, Copenhagen-like, and objective collapse theories. What 
the framework contributes is interpretation-independent: a mathematical correspondence between boundary constraints and observable selection that illuminates the structure of quantum measurement.

\section{Conclusion}

This paper developed three claims connecting boundary conditions, quantum measurement, and stabilizer codes through a first-principles analysis of circuit QED dispersive readout. The first claim is that boundary conditions and measurement interactions share a common structural role: both select a preferred basis by determining which states satisfy the imposed constraint. In a spherical wedge with Dirichlet conditions, sine modes survive while cosine modes are excluded; the operator $L_z$ fails to preserve the constrained domain, while $\partial^2/\partial\phi^2$ remains self-adjoint. In quantum measurement, the interaction Hamiltonian $\Hint$ determines which eigenstates are operationally stable: for dispersive readout with $\Hint = \hbar\chi\sigma_z\hat{n}$, the $\sigma_z$ eigenstates remain distinguishable while superpositions decohere through entanglement with the resonator.

The second claim is that dispersive readout emerges naturally from a first-principles derivation starting from the circuit Lagrangian. The transmission line satisfies the wave equation in the bulk; the transmon at the boundary provides a nonlinear boundary condition derived from current conservation via the variational principle. Linearization for small oscillations yields a frequency-dependent boundary condition, and quantization of the transmon via the Kubo formula produces a state-dependent boundary function $F^{(n)}(\lambda)$ with poles at the qubit transition frequencies. The residues are proportional to squared charge matrix elements and can be written explicitly in terms of circuit parameters as $\delta_{nm} = 2\omega_r|\omega_{nm}||Q_{nm}|^2/(\hbar v Z_0)$. In the dispersive regime, this boundary function becomes approximately rational in the eigenparameter, placing the problem within the Fulton-Walter framework for Sturm-Liouville problems with eigenparameter-dependent boundary conditions. The extended Hilbert space $\mathcal{H}_{\text{ext}} = L^2(0,L) \oplus \mathbb{C}^M$ of the spectral theory provides a mathematical structure: the $L^2$ factor describes the resonator field profile, while $\mathbb{C}^M$ represents boundary amplitudes associated with each transmon transition. This extended space is a mathematical device for the eigenvalue problem with fixed transmon state, distinct from the physical Hilbert space $\mathcal{H}_{\text{res}} \otimes \mathcal{H}_{\text{transmon}}$. Different qubit states define different boundary conditions, hence different resonator frequencies, hence basis selection in dispersive measurement. The dispersive shift $\chi = g^2\alpha/[\Delta(\Delta+\alpha)]$ emerges from solving the transcendental eigenvalue equation $G(\lambda) = F^{(n)}(\lambda)$. The spectral analysis established a level repulsion theorem: no dressed eigenvalue can coincide with a transmon transition frequency, because the boundary function diverges at these points while the resonator function remains finite. This provides a foundation for the avoided crossings observed in circuit QED spectroscopy. At resonance, the vacuum Rabi splitting $2g$ emerges from the eigenvalue equation, matching the Jaynes-Cummings prediction. The multimode extension revealed that the coupling scales as $g_n \propto \sqrt{\omega_n}$, leading to an ultraviolet divergence in the Lamb shift that requires renormalization, while the dispersive shift remains finite.

The third claim concerns multiple qubits and error correction. For two qubits coupled to a shared resonator, the dispersive Hamiltonian commutes with joint parity, making parity a QND observable within the dispersive approximation. Matched dispersive shifts render odd-parity states frequency-degenerate while even-parity states remain distinguishable; true parity-only measurement requires engineered suppression of linear dispersive terms. Stabilizer constraints exhibit the same mathematical structure as boundary conditions, both define subspaces of compatible states, though for $X$-type stabilizers and general stabilizer codes, the physical implementation differs and the correspondence remains structural.

The framework addresses basis selection, not outcome selection. Boundary conditions eliminate incompatible modes from the Hilbert space; measurement correlates branches through entanglement without eliminating any. The contribution is a unified perspective connecting quantum measurement to classical mode selection, grounded in a first-principles derivation from circuit physics, with explicit analysis for dispersive readout and parity measurement in circuit QED.


\end{document}